\documentclass[twoside]{article}

\usepackage
{
        amssymb,
        amsthm,
        amsmath,
        fullpage,
        nicefrac,
        xcolor,
        xspace,
        verbatim
}
\usepackage{booktabs}
\usepackage{mathtools}
\usepackage{framed}
\usepackage{harvard}
\usepackage{ifthen}
\usepackage{algpseudocode,algorithm,algorithmicx}

\algrenewcommand\algorithmicrequire{\textbf{Precondition:}}
\algrenewcommand\algorithmicensure{\textbf{Postcondition:}}
\definecolor{cornellred}{rgb}{0.7, 0.11, 0.11}
\definecolor{dgreen}{rgb}{0.0, 0.5, 0.0}
\definecolor{ballblue}{rgb}{0.13, 0.67, 0.8}
\definecolor{royalblue(web)}{rgb}{0.25, 0.41, 0.88}
\definecolor{bleudefrance}{rgb}{0.19, 0.55, 0.91}
\definecolor{royalazure}{rgb}{0.0, 0.22, 0.66}

\usepackage{cite}
\usepackage{hyperref}
\hypersetup{
    colorlinks = true,
    linkcolor=cornellred,
    citecolor=royalazure,
    linkbordercolor = {white}
}
\usepackage{cleveref}
\newcommand{\Dfun}{\mathcal{F}^{-}}

\newcommand {\compl} {\widetilde}

\newcommand {\Exp}       {\mathbb{E}}

\newcommand {\E}     [1] {\Exp\left[#1\right]}

\newcommand{\given}{\mid}

\newtheorem{theorem}{Theorem}[section]
\newtheorem{lemma}[theorem]{Lemma}
\newtheorem{fact}[theorem]{Fact}
\newtheorem{proposition}[theorem]{Proposition}

\newtheorem{corollary}[theorem]{Corollary}
\newtheorem{definition}[theorem]{Definition}

\newcommand{\ct}{\mathcal T}

\newcommand{\topt}{\ct^*}
\newcommand{\stopt}{\widehat{\ct^*}}
\newcommand{\mw}[1]{\mathcal F^+(#1)}

\DeclareFontFamily{U}{matha}{\hyphenchar\font45}
\DeclareFontShape{U}{matha}{m}{n}{
      <5> <6> <7> <8> <9> <10> gen * matha
      <10.95> matha10 <12> <14.4> <17.28> <20.74> <24.88> matha12
      }{}
\DeclareSymbolFont{matha}{U}{matha}{m}{n}
\DeclareMathSymbol{\wedge}         {2}{matha}{"5E}
\DeclareMathSymbol{\vee}           {2}{matha}{"5F}

\newcommand{\NN}{\mathbb{N}}
\newcommand{\eps}{\varepsilon}

\usepackage[symbol]{footmisc}

\title{Bisect and Conquer: \\Hierarchical Clustering via Max-Uncut Bisection}
\author{\\
Sara Ahmadian\thanks{Google Research NY, New York.}
\and \\
Vaggos Chatziafratis\thanks{Stanford University, California. Supported by an Onassis Foundation Scholarship.}
\and \\
Alessandro Epasto$^*$
\and \\
Euiwoong Lee\thanks{New York University, New York.}
\and \\
Mohammad Mahdian$^*$
\and \\
Konstantin Makarychev\thanks{Northwestern University, Evanston, IL.}
\and \\
Grigory Yaroslavtsev\thanks{Indiana University, Bloomington. Supported by NSF grant 1657477.}
}

\begin{document}
\date{}
\maketitle
\newenvironment{myfont}{\fontfamily{pag}\selectfont}{\par}

\begin{abstract}
\normalsize
Hierarchical Clustering is an unsupervised data analysis method which has been widely used for decades. 
Despite its popularity, it had an underdeveloped analytical foundation and to address this, Dasgupta 
recently introduced an optimization viewpoint of hierarchical clustering with pairwise similarity information that spurred a line of work shedding light on old algorithms (e.g., Average-Linkage), but also designing new algorithms. 
Here, for the maximization dual of Dasgupta's objective (introduced by Moseley-Wang), we present polynomial-time $.4246$ approximation algorithms that use \textsc{Max-Uncut Bisection} as a subroutine.
The previous best worst-case approximation factor in polynomial time was $.336$, improving only slightly over Average-Linkage which achieves $1/3$. 
Finally, we complement our positive results by providing APX-hardness (even for 0-1 similarities), under the \textsc{Small Set Expansion} hypothesis.

\end{abstract}

\section{Introduction}


Hierarchical Clustering (HC) is a popular unsupervised learning method which produces a recursive decomposition of a dataset into clusters of increasingly finer granularity. The output of HC is a hierarchical representation of the dataset in the form of a tree (a.k.a. dendrogram) whose leaves correspond to the data points. 
The internal nodes of the tree correspond to clusters organized in a hierarchical fashion.

Since HC captures cluster structure at all levels of granularity simultaneously, it
offers several advantages over a basic flat clustering (a partition of the data into
a fixed number of clusters) and it has been used for several decades (see e.g., Ward's method~\cite{ward1963hierarchical}). It is one of the most popular methods across a wide range of fields e.g., in phylogenetics~\cite{sneath1962numerical,jardine1968model}, where many of the so-called \textit{linkage-based} algorithms (like Average/Single/Complete-Linkage) originated, in gene expression data analysis~\cite{eisen1998cluster}, the analysis
of social networks~\cite{leskovec2014mining,mann2008use}, bioinformatics~\cite{diez2015novel}, image
and text classification~\cite{steinbach2000}, and even in financial markets~\cite{tumminello2010correlation}. 
See classic texts~\cite{JMF99,MRS08,J10} for a standard introduction to HC and linkage-based methods.
Due to the importance of HC, many variations (including linkage-based methods) are also currently implemented in standard scientific computing packages and in large-scale systems~\cite{bateni2017affinity}.

Despite the plethora of applications, until recently there wasn't a concrete objective associated with HC methods (except for the Single-Linkage clustering which enjoys a simple combinatorial structure due to the connection to minimum spanning trees~\cite{GR69}). This is in stark contrast with flat clustering methods where $k$-means, $k$-median, $k$-center constitute some standard objectives. Recent breakthrough work~\cite{D16} by Dasgupta addressed this gap by introducing the following objective function:

\begin{definition}[HC Objective~\cite{D16}]
Given a similarity matrix with entries $w_{ij} \ge 0$ corresponding to similarities between data points $i$ and $j$, let the hierarchical clustering objective for a tree $\ct$ be defined as:
$$\Dfun(\ct) = \sum_{i < j} w_{ij} |\ct(i,j)|,$$
where $|\ct(i,j)|$ denotes the number of leaves under the least common ancestor of $i$ and $j$.
\end{definition}

The goal of HC under Dasgupta's objective is to minimize the function $\Dfun(\ct)$ among all possible trees. Intuitively, the objective encourages solutions that do not separate similar points (those with high $w_{ij}$) until the lower levels of the tree; this is because $|\ct(i,j)|=n$, if $i,j$ are separated at the top split of the tree $\ct$, whereas $|\ct(i,j)|=2$, if $i,j$ are separated at the very last level.

Dasgupta~\cite{D16} further showed that many desirable properties hold for this objective with respect to recovering ground truth hierarchical clusterings. This was later strengthened both theoretically and experimentally~\cite{cohenaddad,royNIPS}. Although, it is not hard to see that the optimum tree should be binary, it is not clear how one can optimize for it given the vast search space. Surprisingly,~\cite{D16} established a connection with a standard graph partitioning primitive, {\sc Sparsest-Cut},  which had been previously used to obtain HC in practice~\cite{mann2008use}. Later work~\cite{vaggosSODA} further showed a black-box connection: an $\alpha$-approximation for {\sc Sparsest-Cut} or {\sc Balanced-Cut} gives an $O(\alpha)$-approximation for $\Dfun(\ct)$.
Hence an $O(\sqrt{\log n})$-approximation can be computed in polynomial time by using the celebrated result of~\cite{ARV08}. A constant-factor hardness is also known under the \textsc{Small Set Expansion} hypothesis~\cite{vaggosSODA,royNIPS}.

Building on Dasgupta's work, Moseley and Wang~\cite{joshNIPS} gave the first approximation analysis of the  Average-Linkage method. Specifically, for the complement of Dasgupta's objective (see \hyperref[def:MW]{Definition}~\ref{def:MW}), they proved that Average-Linkage achieves a $1/3$ approximation in the worst case. That factor was marginally improved to $.3363$ in~\cite{charikar2019hierarchical} via an ad-hoc semidefinite programming formulation of the problem. This was the state-of-the-art in terms of approximation prior to this work and was arguably complicated and impractical.


\noindent{\bf Contributions:} In this paper, we extend the recent line of research initiated by Dasgupta's work on objective-based hierarchical clustering for similarity data, by significantly beating the previous best-known approximation factor for the~\cite{joshNIPS} objective and doing so with natural algorithms. Our algorithm is based on a combination of \textsc{Max-Uncut Bisection}  and \textsc{Average-Linkage} and guarantees .4246 of the value of the optimum hierarchical clustering, whereas previous work based on semidefinite programming~\cite{charikar2019hierarchical} only achieved .3363. An advantage of our algorithm is that it uses the \textsc{Max-Uncut Bisection} primitive as a black-box and the approximation ratio gracefully degrades as a function of the quality of this primitive; this is in contrast with previous approaches~\cite{charikar2019hierarchical} which solve complicated convex relaxations tailored to the objective.
Since both theoretical and practical solutions for \textsc{Max-Uncut Bisection} are readily available in \cite{austrin,SS13}, this results in a family of algorithms which can be analyzed both rigorously and empirically.
 We also complement our algorithmic results with hardness of approximation (APX-hardness): assuming the \textsc{Small Set Expansion} hypothesis, we prove that even for 0-1 similarities, there exists $\eps >0$, such that it is NP-hard to approximate the ~\cite{joshNIPS} objective within a factor of $(1-\eps)$. A summary of our results compared to the previous work is given in \hyperref[tab:result]{Table}~\ref{tab:result}. Here we also point out that .3333 is a simple baseline achieved by a random binary tree~\cite{joshNIPS} and hence our work gives the first major improvement over this baseline.

\begin{table}[!ht]  
  \centering
  \begin{tabular}{|c||c|c|c|}
    \cline{2-4}
    \multicolumn{1}{c|}{} & \cite{joshNIPS} & \cite{charikar2019hierarchical} & { Our paper} \\ \hline
Approx. & 1/3 & .3363 & {\bf.4246}  \\ \hline
Hardness & NP-hard & - & {\bf APX-hard} \\ &&& under \textsc{SSE}  \\ \hline
  \end{tabular}
  \caption{\label{tab:result}Our results for the~\cite{joshNIPS} objective.}
\end{table}


\noindent{\bf Further Related Work:} 
HC has also been studied in the ``semi-supervised'' or ``interactive'' case, where prior knowledge or expert advice is available, with or without the geometric information on the data points. Examples of these works include~\cite{emamjomeh2018adaptive} where they are interested in the number of (triplet) queries necessary to determine a ground truth HC, and~\cite{dasguptaICML, vaggosICML} who provide techniques for incorporating triplet constraints (the analog of split/merge feedback~\cite{balcan2008clustering, awasthi2017local} or must-link/cannot-link constraints~\cite{wagstaff2001constrained,wagstaff2000clustering} in standard flat clustering) to get better hierarchical trees. Furthermore, assuming the data points lie in a metric space (instead of $w_{ij}$ being arbitrary similarities), there are previous works that measure the quality of a HC using standard flat-clustering objectives like $k$-means, $k$-median or $k$-center as proxies~\cite{charikar2004incremental,dasgupta2002performance,plaxton2006approximation,lin2010general} in order to get approximation guarantees. Finally, in high dimensions when even running simple algorithms like single-linkage becomes impractical due to the exponential dependence
on the dimension~\cite{yaroslavtsev2017massively}, one can still efficiently achieve good approximations if the similarities are generated with the fairly common Gaussian Kernel and other smooth similarity measures, as is shown in~\cite{charikar2018hierarchical}.

\noindent{\bf Organization:} We start by providing the necessary background in \hyperref[sec:prelim]{Section}~\ref{sec:prelim}. Then, we give our .4246 approximation algorithm for HC based on \textsc{Max-Uncut Bisection} in \hyperref[sec:improved]{Section}~\ref{sec:improved} and our hardness of approximation result in \hyperref[sec:hardness]{Section}~\ref{sec:hardness}. Our conclusion is given in \hyperref[sec:conclusion]{Section}~\ref{sec:conclusion}.

\section{Preliminaries}
\label{sec:prelim}

We begin by setting the notation used throughout the paper. The input to the HC problem is given by an (explicit or implicit) $n \times n$ matrix of pairwise similarities with non-negative entries $w_{ij}\ge0$ corresponding to the similarity between the $i$-th and $j$-th data point. Sometimes we also refer to the underlying weighted graph as $G(V,E,w)$.

 We denote by $\ct$ a rooted tree whose leaves correspond to the $n$ vertices. For two leaves $i,j$, $\ct(i,j)$ denotes the subtree rooted in the least common ancestor of $i$ and $j$ in $\ct$ and $|\ct(i,j)|$ the number of leaves contained in $\ct(i,j)$. 

 For a set $S \subseteq V$, $w(S) \coloneqq \sum_{(i < j) \in S \times S} w_{ij}$ denotes the total weight of pairwise similarities inside $S$. For a pair of disjoint sets $(S,T) \in V \times V$, $w(S,T) \coloneqq \sum_{i \in S, j \in T} w_{ij}$ denotes the total weight of pairwise similarities between $S$ and $T$.

\begin{definition}[HC Objective~\cite{joshNIPS}]
\label{def:MW}
For a tree $\ct$ consider the hierarchical clustering objective:
\begin{equation}\label{eq:obj}
\mw{\ct} = \sum_{1 \le i < j \le n} w_{ij} (n -|\ct(i,j)|)\tag{*}
  \end{equation}
\end{definition}
If we denote the total weight of the graph by $W\coloneqq \sum_{i < j}w_{ij}$, then note that a trivial upper bound on the value of this objective is $(n-2)W$.

\textsc{Average-Linkage}: One of the main algorithms used in practice for HC, it starts by merging clusters of data points that have the highest average similarity. It is known that it achieves $1/3$ approximation for the Moseley-Wang objective and this is tight in the worst case\footnote{As is common in worst-case analysis of algorithms, the $1/3$ tightness for Average-Linkage is based on a rather brittle pathological counterexample, even though its performance is much better in practice (see for example~\cite{charikar2018hierarchical}).}~\cite{charikar2019hierarchical}. For a formal description, please refer to \hyperref[algo:average]{Algorithm}~\ref{algo:average}.

\textsc{Max-Uncut Bisection}: This is the complement problem to \textsc{Min-Cut Bisection} (which is perhaps more standard in the literature), and the goal here is to split the vertices of a weighted graph into two sets $(S,\bar{S})$, such that the weight of uncut edges $\sum_{ij\in E}w_{ij} - \sum_{i\in S,j\in \bar{S}}w_{ij}$ is maximized. It is known that  one can achieve at least $.8776$ of the optimum value in polynomial time~\cite{austrin,wu2015improved}.


\begin{algorithm}[h]
\caption{\textsc{Average-Linkage}}
\label{algo:average}
\begin{algorithmic}[1]
\State \textbf{input:} Similarity matrix $w \in \mathbb R_{\ge 0}^{n \times n}$.
\State Initialize clusters $\mathcal{C}  \leftarrow \cup_{v\in V}\{v\} $.
\While {$|\mathcal{C}|\ge2$}
\State  Pick $A,B \in \mathcal{C}$ to maximize: \\ \qquad \qquad $w(A,B) \coloneqq \tfrac{1}{|A||B|}\sum_{a\in A,b\in B}w_{ab}$
\State Set $\mathcal{C}\leftarrow \mathcal{C}\cup\{A\cup B\}\setminus\{A,B\}$

\EndWhile
\end{algorithmic}
\end{algorithm}
\begin{algorithm}[h]
\caption{HC via \textsc{Max-Uncut Bisection}}
\label{algo:uncut}
\begin{algorithmic}[1]
\State \textbf{input:} Similarity matrix $w \in \mathbb R_{\ge 0}^{n \times n}$.
\State Partition the underlying graph on $n$ vertices with edges weighted by $w$ into two parts $S$ and $\bar{S}$ using \textsc{Max-Uncut Bisection} as a black box. This creates the top split in the hierarchical clustering tree.
\State Run \textsc{Average-Linkage} on $S$ and on $\bar{S}$ to get trees $\mathcal{T_{S}}$ and $\mathcal{T_{\bar{S}}}$.
\State Construct the resulting HC tree by first splitting into  $(S,\bar{S})$, then building trees $\mathcal{T_{S}}$ and $\mathcal{T_{\bar{S}}}$ on the respective sets.
\end{algorithmic}
\end{algorithm}



\section{0.4246 approximation for HC}
\label{sec:improved}
Our algorithm, based on \textsc{Max-Uncut Bisection} and \textsc{Average-Linkage}, is simple to state and is given in \hyperref[algo:uncut]{Algorithm}~\ref{algo:uncut}. It starts by finding an approximate solution to \textsc{Max-Uncut Bisection}, followed by \textsc{Average-Linkage} agglomerative hierarchical clustering in each of the two pieces\footnote{With no change in approximation one can also run \textsc{Max-Uncut Bisection} recursively, however running \textsc{Average-Linkage} is typically substantially more efficient.}. Our main result is that this approach produces a binary tree which achieves $.4246$ of the optimum:

\begin{theorem}\label{thm:main}
Given an instance of hierarchical clustering, our \hyperref[algo:uncut]{Algorithm}~\ref{algo:uncut} outputs a tree achieving $\frac{4\rho}{3 (2\rho + 1)}-o(1) \ge .4246$ (for $\rho = 0.8776$) of the optimum according to the objective (\ref{eq:obj}), if a $\rho$-approximation for the \textsc{Max-Uncut Bisection} problem is used as a black-box.
\end{theorem}
\noindent{\bf Remark:} The current best approximation factor achievable for \textsc{Max-Uncut Bisection} in polynomial time is $\rho=0.8776$. This makes our analysis almost tight, since one can't get better than $.444$ even by using an exact \textsc{Max-Uncut Bisection} algorithm (with $\rho = 1$).


\subsection{Overview of the proof}
Before delving into the technical details of the main proof, we present our high-level strategy through a series of 4 main steps:
\paragraph{Step 1:} Consider a binary\footnote{W.l.o.g. the optimal tree can be made binary.} tree $\topt$ corresponding to the optimal solution for the hierarchical clustering problem and let $\texttt{OPT} = \mw{\topt}$ be the value of the objective function for this tree. Note that there exists a subtree $\stopt$ in this tree which contains more than $n/2$ leaves while its two children contain at most $n/2$ leaves each (see \hyperref[fig:tree]{Figure}~\ref{fig:tree}). Given this decomposition of the optimum tree into three size restricted sets $A,B,C$, we provide an upper bound for $\texttt{OPT}$ as a function of the weight inside and across these sets (see \hyperref[prop:opt-ub]{Proposition}~\ref{prop:opt-ub}). We then need to do a case analysis based on whether the weight across or inside these sets is larger.

\paragraph{Step 2:}
In the former case, things are easy as one can show that $\texttt{OPT}$ is \textit{small} and that even the contribution from the \textsc{Average-Linkage} part of our algorithm alone yields a $\frac49$-approximation. This is carried out in \hyperref[cor:beta>alpha]{Proposition}~\ref{cor:beta>alpha} based on the \hyperref[prop:rand]{Fact}~\ref{prop:rand}.

\paragraph{Step 3:}
In the latter case, we show that there exists a split of the graph into two exactly equal pieces, so that the weight of the uncut edges is relatively \textit{large}. This is crucial in the analysis as having a good solution to the \textsc{Max-Uncut Bisection} directly translates into a high value returned by the $\rho$-approximate black box algorithm (see \hyperref[lem:uncut-lb]{Lemma}~\ref{lem:uncut-lb}, \hyperref[prop:uncut-red]{Proposition}~\ref{prop:uncut-red} and~\hyperref[prop:uncut-blue]{Proposition}~\ref{prop:uncut-blue}).

\paragraph{Step 4:}
Finally, from the previous step we know that the returned value of the black box is large, hence taking into account the form of the HC objective, we can derive a lower bound for the value our \hyperref[algo:uncut]{Algorithm}~\ref{algo:uncut}. The proof of the main theorem is then completed by \hyperref[prop:sdp-bound]{Proposition}~\ref{prop:sdp-bound} and \hyperref[lem:final]{Lemma}~\ref{lem:final}.

\begin{figure}
\center
  \includegraphics[width=10cm]{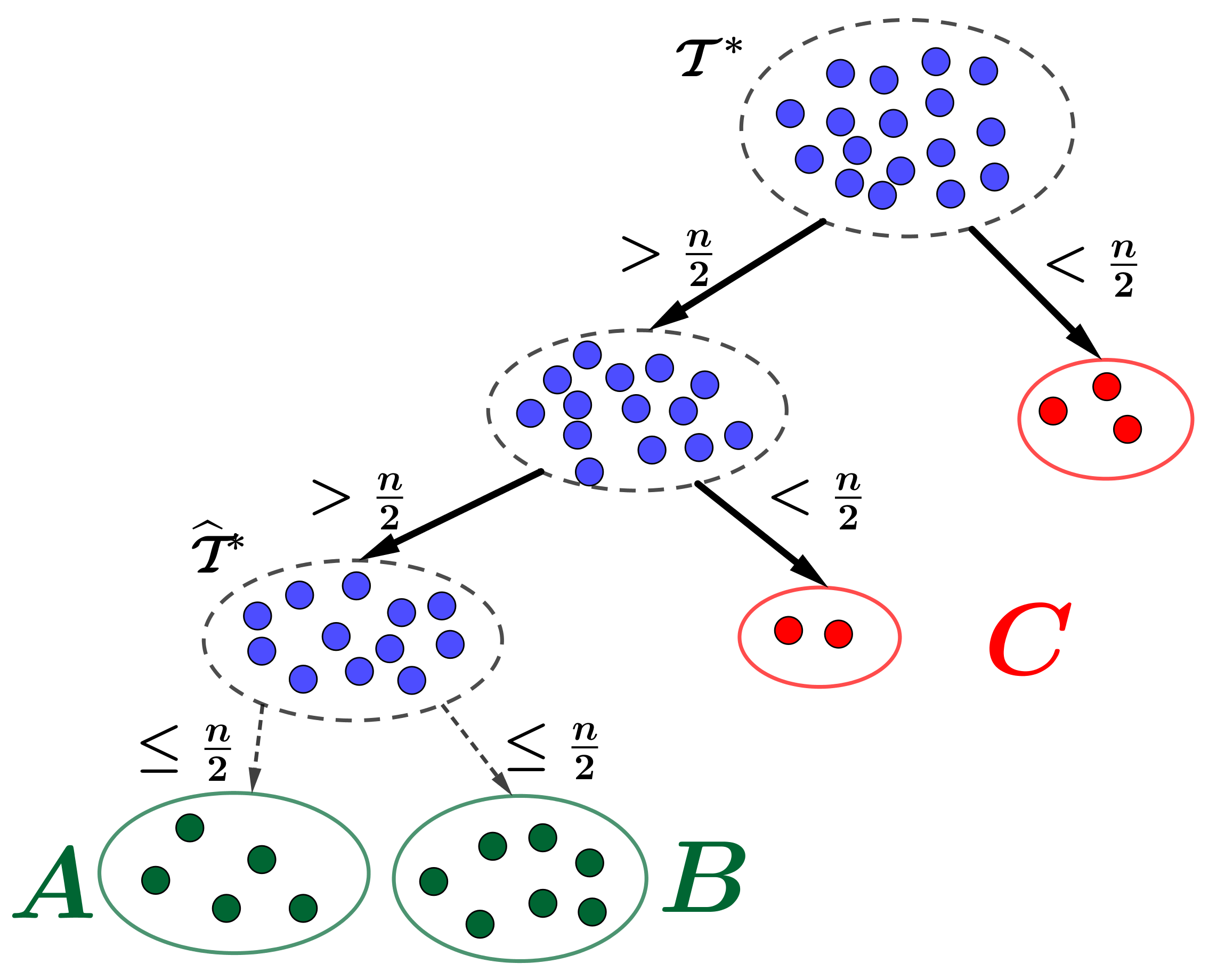}
  \caption{Splitting $\topt$ to size restricted sets $A,B,C$.}
  \label{fig:tree}
\end{figure}

\subsection{Proof of Theorem~\ref{thm:main}}

For ease of presentation, we assume $n$ is even, to avoid the floor/ceiling notation and we omit the $o(1)$ terms.

\begin{proposition}\label{prop:opt-ub}
Let $A$ be the set of leaves in the left subtree of $\stopt$, let $B$ the set of  leaves in the right subtree of $\stopt$ and $C = V \setminus (A \cup B)$ be the set of leaves outside of $\stopt$. Then\footnote{By definition, $A,B$ contain at most $\tfrac{n}{2}$ leaves, while $C$ contains strictly fewer than $\tfrac{n}{2}$ leaves.}: 
\[\texttt{OPT} \leq \left(w(A) + w(B) + w(C)\right)\cdot(n-2) +\] 
\[+\left(w(A,B) + w(B,C) + w(A,C)\right)\cdot|C|\]
\end{proposition}



\begin{proof} For an edge $(i,j)$ whose endpoints lie in the same cluster (i.e., $A$, $B$ or $C$), its contribution to the objective is at most $w_{ij}(n-2)$, using the trivial upper bound of $(n-2)$. Consider any pair of leaves $(i,j) \in A \times B$ in $\topt$. The least common ancestor for this pair is the root of $\stopt$  and hence the contribution of this pair to the objective is equal to $w_{ij} (n - |\stopt|) = w_{ij}|C|$. Similarly, for any pair of leaves $(i,j) \in A \times C$ (or in $B \times C$), their least common ancestor is a predecessor of the root of $\stopt$ and hence the contribution of this pair to the objective is at most $w_{ij} (n - |\stopt|) = w_{ij}|C|$. The desired bound now follows by summing up all the contributions of all distinct pairs of leaves.
\end{proof}

From now on, let $\alpha \coloneqq w(A) + w(B) + w(C)$ and let $\beta \coloneqq w(A,B) + w(B,C) + w(A,C)$ denote the total weights of similarities inside the three sets and crossing a pair of these sets respectively. Note the total weight of all similarities is $W=\alpha+\beta$. 

\begin{fact}[\textsc{Average-Linkage}~\cite{joshNIPS}]\label{prop:rand}
The \textsc{Average-Linkage} algorithm gives a solution whose $\mathcal F^+$ objective is at least $\frac13W(n-2)=\frac13(\alpha+\beta)(n-2)$.
\end{fact}


\begin{proposition}\label{cor:beta>alpha}
If $\alpha \le \beta$, our \hyperref[algo:uncut]{Algorithm}~\ref{algo:uncut} outputs a solution of 
value at least $4/9 \texttt{OPT} \geq 0.44 \texttt{OPT}$, where \texttt{OPT} denotes the HC value of any optimum solution.
\end{proposition}
\begin{proof}
Recall that by definition of $C$ it holds that $|C| < \frac n2\le\frac n2 -1\le\frac {n-2}{2}$. Hence by \hyperref[prop:opt-ub]{Proposition}~\ref{prop:opt-ub} we have $\texttt{OPT} \le \alpha(n-2) + |C|\cdot \beta \le \alpha(n-2) + \frac{n-2}{2} \beta$.
On the other hand, by \hyperref[prop:rand]{Fact}~\ref{prop:rand}, 
Average-Linkage outputs a solution whose expected value is $\frac13 (\alpha + \beta)(n-2)$. We have $\frac13 (\alpha + \beta)(n-2) - \frac49 (\alpha(n-2) + \frac {n-2}{2} \beta) = \frac19 (\beta - \alpha) \ge 0$.
Hence, the 
Average-Linkage part of the algorithm alone gives a $\frac49$-approximation in this case.
\end{proof}

\begin{lemma}\label{lem:uncut-lb}
Suppose $\alpha \geq \beta$. Then, there exists a balanced cut $(L,R)$ of the nodes in $G$, such that the weight of the uncut edges is at least
$\alpha - (\alpha-\beta)\delta_{max}(c)$, where $c=|C|/n$ and $\delta_{max}(c) = \tfrac{c(1-2c)}{(1-3c^2)}$.
\end{lemma}
\begin{proof}
For the partition $(A, B, C)$ we will refer to edges whose both endpoints are inside one of the three sets as \textit{red} edges (i.e., $(i,j) \in (A \times A) \cup (B \times B) \cup (C \times C)$).
We refer to the edges whose two endpoints are contained in two different sets as \textit{blue} edges (i.e., $(i,j) \in (A \times B) \cup (A \times C) \cup (B \times C)$). Our goal here is to give a randomized partitioning scheme that produces the bisection $(L,R)$ with high value of uncut weight lying inside $L,R$.

For simplicity, recall that $n$ is even. The case of odd $n$ is handled similarly.
Denote $a = |A|/n$ and $b=|B|/n$. Let $\compl a = 1/2-a$, $\compl b = 1/2-b$, and $\compl c = 1/2-c$. Note that $\compl a$, $\compl b$ are non-negative, and $\compl c$
is strictly positive due to the size restrictions. Define:
$$
q_A = \frac{2\compl b \compl c}{(\compl b + \compl c)^2},\;\;
q_B = \frac{2\compl a \compl c}{(\compl a + \compl c)^2},\;\;
q_C = \frac{2\compl a \compl b}{(\compl a + \compl b)^2},$$
and
$$
p_A = \frac{q_B q_C}{q_A q_B + q_B q_C + q_A q_C},$$
$$
p_B = \frac{q_A q_C}{q_A q_B + q_B q_C + q_A q_C},$$
$$
p_C = \frac{q_A q_B}{q_A q_B + q_B q_C + q_A q_C}.
$$
We also denote the following expression by $\delta$:
$$\delta = \frac{q_A\, q_B\, q_C}{q_A q_B + q_B q_C +q_A q_C}.$$

Consider the following partitioning procedure:
\begin{itemize}
\item Pick one of the sets $A$, $B$, or $C$ with probability $p_A$, $p_B$, and $p_C$, respectively (note that
 $p_A + p_B + p_C = 1$).
\item If the chosen set is $A$, partition it into two random sets $S_B$ and $S_C$ of size
$\compl b |A|/(\compl b + \compl c)$ and $\compl c |A|/(\compl b + \compl c)$ and output the cut $L= B\cup S_B$, $R = C \cup S_C$.
\item Similarly, if the chosen set is $B$, we partition it into two random sets $S_A$ and $S_C$ of size
$\compl a |B|/(\compl a + \compl c)$ and $\compl c |B|/(\compl a + \compl c)$ and output the cut $L= C\cup S_C$, $R = A \cup S_A$.
\item If the chosen set is $C$, we partition it into two random sets $S_A$ and $S_B$ of size
$\compl a |C|/(\compl a + \compl b)$ and $\compl b |C|/(\compl a + \compl b)$ and output the cut $L= A\cup S_A$, $R = B \cup S_B$.
\end{itemize}
We first observe that each of the output sets $L$ and $R$ has $n/2$ vertices, i.e., $(L,R)$ is a bisection of the graph. If for instance, the algorithm picks
set $A$ at the first step, then the set $L$ contains $|B| + \compl b |A|/(\compl b + \compl c)$ vertices. We have
$$|L|= |B| + \frac{\compl b}{\compl b + \compl c} |A| = bn + \frac{1/2-  b}{1 - b - c} \cdot an=$$
$$= bn + \frac{1/2-b}{a} \cdot an = \frac{n}{2}.$$
The set $R$ is the complement to $L$, thus, it also contains $n/2$ vertices. The cases when the algorithm picks the set $B$ or $C$ are
identical.

We now compute the expected weight of red edges in the bisection $(L,R)$. 

\begin{proposition}\label{prop:uncut-red}
The expected weight of uncut red edges is $(1 - \delta) \alpha$.
\end{proposition}
\begin{proof}
Again, assume that the algorithm picks the set $A$ at the first step.
Then, the sets $B$ and $C$ are contained in the sets $L$ and $R$, respectively. Consequently, no edges in $B$ and $C$ are in the cut between
$L$ and $R$. Every edge in $A$ is cut with probability $2\compl b \compl c/(\compl b + \compl c)^2$. Thus, the weight
of red edges in the cut between $L$ and $R$ (denoted as $E^{red}(L,R)$) given that the algorithm picks set $A$ equals
$$\E{|E^{red}(L,R)|\given\text{ algorithm picks $A$ at first step}}=$$
$$=\frac{2\compl b \compl c}{(\compl b + \compl c)^2} w(A) = q_A w(A).$$
Similarly, if the algorithm picks the set $B$ or $C$, the expected sizes of the cuts equal
$q_B w(B)$ and $q_C w(C)$,
respectively. Hence, the expected weight of the red edges between $L$ and $R$ (when we do not condition on the first step of the algorithm) equals
$$
\E{|E^{red}(L,R)|}=p_Aq_Aw(A)+p_Bq_B w(B) + p_C\, q_Cw(C)$$
Observe that 
$$p_Aq_A = p_Bq_B=p_Cq_C = \frac{q_A\, q_B\, q_C}{q_A q_B+q_B q_C+q_A q_C} = \delta.$$ 

Then, the expected weight of red edges between $L$ and $R$ equals:
$$
\E{|E^{red}(L,R)|} = \delta (w(A)+w(B)+w(C))=\delta \alpha.
$$
Here, we used that $w(A)+w(B)+w(C) = \alpha$.
The expected weight of uncut red edges equals $(1-\delta)\alpha$.
\end{proof}

We now lower bound the weight of uncut blue edges.
\begin{proposition}\label{prop:uncut-blue}
The expected weight of uncut blue edges is at least $\delta \beta$.
\end{proposition}
\begin{proof}
We separately consider edges between sets $A$ and $B$, $B$ and $C$, $A$ and $C$. 
Consider an edge $(u,v)\in A\times B$. This edge is not in
the cut $(L,R)$ if both endpoints $u$ and $v$ belong to $L$ or both
endpoints belong to $R$. The former event -- $\{u,v\in L\}$ -- 
occurs if the set $B$ is chosen in the first step of the
algorithm and the set $S_A$ contains vertex $v$; the latter event -- $\{u,v\in R\}$ -- occurs if $A$ is chosen in the first step of
the algorithm and the set $S_B$ contains vertex $u$. The probability 
of the union of these events is\footnote{Note that $\compl c = 0$ cannot happen by definition of the sets $A,B,C$.}
$$
\Pr[(u,v)\notin (L,R)] = 
p_B \cdot \frac{\compl{a}}{\compl a+\compl c}
+p_A \cdot \frac{\compl b}{\compl b+\compl c}
=$$
$$=p_B q_B \cdot \frac{(\compl a + \compl c)}{2\compl c}
+p_A q_A \cdot \frac{(\compl b + \compl c)}{2\compl c}.
$$
Since $p_Aq_A=p_Bq_B=\delta$, we have
$$\Pr[(u,v)\notin (L,R)] = 
\delta \cdot \frac{\compl a+\compl b + 2\compl c}{2\compl c} = 
\delta \cdot \frac{1/2 + \compl c}{2\compl c} \geq \delta.
$$
The last inequality holds because $(1/2+\compl c)/\compl c\geq 2$ for
all $\compl c\in(0,1/2]$. The same bound holds for edges 
between sets $B$ and $C$ and sets $A$ and $C$. Therefore,
the expected weight of uncut blue edges is at least 
$\delta \beta$.
\end{proof}

By the above two propositions (\hyperref[prop:uncut-red]{Proposition}~\ref{prop:uncut-red} and \hyperref[prop:uncut-blue]{Proposition}~\ref{prop:uncut-blue})
the expected total weight of uncut edges is at least: 
$$
\E{w(L) + w(R)} \geq (1-\delta)\alpha + \delta \beta =
\alpha - (\alpha - \beta)\delta.
$$
Note that we are in the case with $\alpha -\beta \geq 0$. Thus to establish 
a lower bound on the expectation, we need to show an upper bound on 
$\delta$. Write

$$
\delta = \frac{q_A\, q_B\, q_C}{q_A q_B + q_B q_C + q_A q_C} 
       = \frac{1}{\frac{1}{q_A} + \frac{1}{q_B} + \frac{1}{q_C}}.
$$
After plugging in the values of $q_A$, $q_B$, and $q_C$, we obtain the following expression for $\delta$.
$$
\delta = \frac{1}{
\frac{(\compl b + \compl c)^2}{2\compl b \compl c}+
\frac{(\compl a + \compl c)^2}{2\compl a \compl c}+
\frac{(\compl a + \compl b)^2}{2\compl a \compl b}}
= $$
$$=\frac{1}{
3 + \frac{1}{2}\big(
\frac{\compl b + \compl c}{\compl a} +
\frac{\compl a + \compl c}{\compl b} +
\frac{\compl a + \compl b}{\compl c}
\big)}
$$
Observe that $a+b+c =1$ and $\compl a + \compl b + \compl c = 1/2$. Thus, $\compl b + \compl c = 1/2 - \compl a$,
$\compl a + \compl c = 1/2 - \compl b$, and $\compl a + \compl b = 1/2 - \compl c$. Hence,
$$\delta = \frac{1}{
3 + \frac{1}{2}\big(
\frac{1/2 - \compl a}{\compl a} +
\frac{1/2 - \compl b}{\compl b} +
\frac{1/2 - \compl c}{\compl c}
\big)}
=$$
$$=
\frac{1}{
\frac{3}{2} + \frac{1}{4}\big(
\frac{1}{\compl a} +
\frac{1}{\compl b} +
\frac{1}{\compl c}
\big)}.
$$
Note that since the function $t \mapsto 1/t$ is convex for $t > 0$, we have
$$\frac{1}{2}\Big(\,\frac{1}{\compl a} + \frac{1}{\compl b}\,\Big) \geq \frac{2}{\compl a + \compl b} = \frac{2}{1/2 - \compl c} = \frac{2}{c}$$
Therefore,
$$\delta\leq
\frac{1}{\frac{3}{2} + \frac{1}{c} + \frac{1}{4\compl c}}  = \frac{c(1-2c)}{1-3c^2}.$$
We conclude that the expected weight of uncut edges is at least 
$\alpha - (\alpha - \beta)\delta_{max}(c)$, where 
$\delta_{max}(c) = c(1-2c)/(1-3c^2)$.
\end{proof}

\begin{proposition} \label{prop:sdp-bound}
Let $\rho = 0.8776$ be the approximation 
factor of the \textsc{Max-Uncut Bisection} algorithm~\cite{austrin,wu2015improved}. 
Then if $\beta \le \alpha$, our \hyperref[algo:uncut]{Algorithm}~\ref{algo:uncut} outputs a solution of value at least 
$\tfrac{2\rho}{3}(n-1)((1 - \delta_{max}(c))\alpha + \delta_{max}(c) \beta)$.
\end{proposition} 
\begin{proof}
Let $(L,R)$ be the bisection produced by the $\rho$-approximate \textsc{Max-Uncut Bisection} algorithm. This partition satisfies:
$$w(L) + w(R) \ge \rho \texttt{OPT}_{\textsc{Max-Uncut Bisection}}$$
Our \hyperref[algo:uncut]{Algorithm}~\ref{algo:uncut} produces a tree where at the top level the left subtree is $L$, the right subtree is $R$ and both of these subtrees are then generated by \textsc{Average-Linkage}. Hence each edge $(i,j) \in L \times L$ (and similarly for edges in $R \times R$) contributes: 
$$w_{ij}\left(\tfrac n2 + \tfrac13(\tfrac n2 - 2)\right) =\tfrac23w_{ij} (n-1) $$
to the objective. Thus the overall value of our solution is at least:
$$\tfrac23 (n-1)(w(L) + w(R))\ge$$
$$\ge\frac{2\rho}{3} (n -1)\cdot\texttt{OPT}_{\textsc{Max-Uncut Bisection}}$$

If $\beta \le \alpha$ then by \hyperref[lem:uncut-lb]{Lemma}~\ref{lem:uncut-lb} 
we have that $\texttt{OPT}_{\textsc{Max-Uncut Bisection}} \ge \alpha - (\alpha - \beta)\delta_{max}(c)$ and the proof follows by rearranging the terms.
\end{proof}


\begin{lemma}\label{lem:final}
The approximation factor $\xi$ of our \hyperref[algo:uncut]{Algorithm}~\ref{algo:uncut}  is at least $\frac{4\rho}{3 (2\rho + 1)} \ge 0.42469$.
\end{lemma}
\begin{proof}
First, note that if $\beta \ge \alpha$ then by \hyperref[cor:beta>alpha]{Proposition}~\ref{cor:beta>alpha} the approximation is at least $0.44$. Hence it suffices to only consider the case when $\beta \le \alpha$. 
Recall that by \hyperref[prop:rand]{Fact}~\ref{prop:rand}, \textsc{Average-Linkage} outputs a solution of value $\frac13(\alpha + \beta)(n-2)$ and by \hyperref[prop:opt-ub]{Proposition}~\ref{prop:opt-ub}, we have $\texttt{OPT} \le \alpha(n-2) + |C| \beta$.
Hence if $\frac13(\alpha + \beta)(n-2) \ge \xi\left(\alpha(n-2) + |C| \beta\right)$ then the desired approximation holds.

Thus we only need to consider the case when $\frac13(\alpha + \beta)(n-2) \le \xi\left(\alpha(n-2) + |C|\beta\right)$ or equivalently: 
$$\frac13(\alpha + \beta) \le \xi\left(\alpha + \tfrac{|C|}{n-2}\beta\right)\iff$$
$$\iff\beta \le \frac{3 \xi - 1}{1 - 3 \tfrac{|C|}{n-2} \xi} \alpha.$$
Let $c_1 = \tfrac{2\rho}{3} (1 - \delta_{max}(c))$ and $c_2 = \tfrac{2\rho}{3} \delta_{max}(c)$.
In this case by \hyperref[prop:sdp-bound]{Proposition}~\ref{prop:sdp-bound}, our \hyperref[algo:uncut]{Algorithm}~\ref{algo:uncut}  gives value at least $c_1(n-1) \alpha + c_2(n-1) \beta$.
Hence it suffices to show that $c_1(n-2) \alpha + c_2 (n-2)\beta \ge \xi (\alpha(n-2) + |C| \beta)$.
Or equivalently that:
$$\beta (\tfrac{|C|}{n-2} \xi - c_2) \le \alpha (c_1 - \xi)$$
Using the bound on $\beta$ above it suffices to show that:
$$\frac{3 \xi - 1}{1 - 3 \tfrac{|C|}{n-2} \xi} (\tfrac{|C|}{n-2} \xi - c_2) \le c_1 - \xi.$$
After simplifying this expression the above bound holds for:
$$\xi \le \frac{c_1 - c_2}{1 + 3 \tfrac{cn}{n-2} c_1 - \tfrac{cn}{n-2} - 3 c_2}.$$
Hence it suffices to find the minimum of the RHS over $c \in [0,\tfrac12-\tfrac1n]$.
Plugging in the expressions for $c_1$ and $c_2$ after simplification the RHS is equal to:
$$\frac23 \frac{\rho(1 - c)}{2 c^2 \rho + (1 - 3c^2)}$$

Differentiating over $c$ one can show that the minimum of this expression is attained for $c = \tfrac12-\tfrac1n$. Indeed, the numerator of the derivative is quadratic function with negative leading coefficient whose roots are $1 \pm \frac{\sqrt{t(t + 1)}}{t}$ for $t = 2 \rho - 3$. The left root is approximately $0.545$ and hence the derivative is negative on $[0,\tfrac12-\tfrac1n]$.
The value at the minimum $c = \tfrac12-\tfrac1n$ is thus equal\footnote{Observe that our analysis based on \textsc{Max-Uncut Bisection} is almost tight since even if we were given exact access to the optimum (i.e., $\rho=1$), the approximation ratio for HC would only slightly increase to $\tfrac{4}{9}=.444$.} to ($\rho=0.8776$):
$$\frac{4\rho}{3 (2\rho + 1)} \ge 0.42469$$
\end{proof}

\section{Hardness of Approximation}
\label{sec:hardness}
In this section, we prove that maximizing the Moseley-Wang HC objective~(\ref{eq:obj})  is APX-hard: 
\begin{theorem}\label{th:hardness}
Under the \textsc{Small Set Expansion} (\textsc{SSE}) hypothesis, there exists $\eps > 0$, such that it is NP-hard to approximate the Moseley-Wang HC objective function~(\ref{eq:obj})  within a factor $(1 - \eps)$. 
\end{theorem}
Initially introduced by Raghavendra and Steurer~\cite{raghavendra2010graph}, \textsc{SSE} has been used to prove improved hardness results for optimization problems including \textsc{Balanced Separator} and \textsc{Minimum Linear Arrangement}~\cite{raghavendra2012reductions}.

Given a $d$-regular, unweighted graph $G = (V, E)$ and $S \subseteq V$, let $\mu(S) := |S| / |V|$ and $\Phi(S) := |E(S, V \setminus S)| / d|S|$. Raghavendra et al.\cite{raghavendra2012reductions} prove the following strong hardness result.
(While it is not explicitly stated that the result holds for regular graphs, it can be checked that their reduction produces a regular graph~\cite{Tulsiani19}.)

\begin{theorem}[Theorem 3.6 of \cite{raghavendra2012reductions}]
Assuming the \textsc{SSE}, for any $q \in \NN$ and $\eps, \gamma > 0$, given a regular graph $G = (V, E)$, it is NP-hard to distinguish the following two cases.

\begin{itemize}
\item YES: There exist $q$ disjoint sets $S_1, \dots, S_q \subseteq V$ such that for all $\ell \in [q]$, 
\[
\mu(S_{\ell}) = 1/q \qquad \mbox{and} \qquad \Phi(S_{\ell}) \leq \eps + o(\eps).
\]
\item NO: For all sets $S \subseteq V$, 
\[
\Phi(S) \geq \phi_{1-\eps/2}(\mu(S)) - \gamma / \mu(S)
\]
where $\phi_{1-\eps/2}(\mu(S))$ is the expansion of the sets of volume $\mu(S)$ in the infinite Gaussian graph with correlation $1 - \eps/2$. 
\end{itemize}
\label{thm:rst}
\end{theorem}

\begin{proof}[Proof of Theorem~\ref{th:hardness}]
Let us consider the instance of Hierarchical Clustering defined by the same graph where each pair has weight 1 if there is an edge, and 0 otherwise. Then $W = |E|$ is the total weight.

\begin{itemize}
\item YES: The fraction of edges crossing between different $S_i$'s is at most $\eps + o(\eps)$, and all edges inside some $S_i$ are multiplied by at least $n(1 - 1/q)$ in the objective function. So the objective function for Hierarchical clustering is at least 
\[
(1 - \eps - o(\eps))W\cdot(1 - \tfrac{1}{q})n \geq nW (1 - \tfrac1q - \eps - o(\eps)).
\]

\item NO: Consider an arbitrary binary tree $\ct$ that maximizes the Moseley-Wang objective function~(\ref{eq:obj}). 
For a tree node $a \in \ct$, let $\ct_a$ be the subtree of $\ct$ rooted at $a$, and $V_a \subseteq V$ be the set of graph vertices corresponding the leaves of $T_a$. 

Let $b \in \ct$ be a highest node such that $n/3 \leq |V_b| \leq 2n/3$ (such a node always exists in a binary tree). By \hyperref[thm:rst]{Theorem}~\ref{thm:rst}, we have
\[
\Phi(V_b) \geq \phi_{1-\eps/2}(\mu(V_b)) - \gamma / \mu(V_b) \geq C \sqrt{\eps}
\]
for some absolute constant $C$. Here we use the fact that 
\[
\phi_{1-\eps/2}(\mu(V_b)) \geq \Omega(\sqrt{\eps}) \mbox{ for } \mu(V_b) \in [1/3, 2/3]
\]
and take $\gamma$ small enough depending on $\eps$.

So the total fraction of edges in $E(V_b, V \setminus V_b)$ is at least 
\[
\mu(V_b) \cdot \Phi(V_b) \geq \tfrac{C\sqrt{\eps}}{3}. 
\]
Note that edges in $E(T_b, V \setminus T_b)$ will multiplied by at most $n /3$ in the objective function. 
(Let $a$ be the parent of $b$. Then $|V_a| > 2n / 3$ by the choice of $b$ and for any edge crossing $V_b$, the the least common ancestors of the two endpoints will be $a$ or one if its ancestors.) 
Therefore, the objective function is at most 
\[
nW - \tfrac{C\sqrt{\eps}}{3}W \cdot \tfrac{2n}{3} = nW(1 - \tfrac{2 C \sqrt{\eps}}{9}).
\]
\end{itemize}

Therefore, the value is at least $nW(1 - 1/q - \eps - o(\eps))$ in the YES case and 
$nW(1 - (2C \sqrt{\eps}/9))$ in the NO case. By taking $\eps > 0$ sufficiently small and $q$ arbitrarily large, there is a constant gap between the YES case value and the NO case value. 
\end{proof}

\section{Conclusion}
\label{sec:conclusion}

In this paper, we presented a .4246 approximation algorithm for the hierarchical clustering problem with pairwise similarities under the Moseley-Wang objective~(\ref{eq:obj}), which is the complement to Dasgupta's objective. Our algorithm uses  \textsc{Max-Uncut Bisection} as a black box and improves upon previous state-of-the-art approximation algorithms that were more complicated and only guaranteed .3363 of the optimum value. In terms of hardness of approximation, under the \textsc{Small Set Expansion} hypothesis, we prove that even for unweighted graphs, there exists $\eps >0$, such that it is NP-hard to approximate the objective function~(\ref{eq:obj})  within a factor of $(1-\eps)$.

\bibliographystyle{alpha}
\bibliography{main.bib}

\end{document}